\documentclass[aps,pra,showpacs,twoside,10pt,longbibliography,twocolumn]{revtex4-1}
\usepackage[colorlinks=true, citecolor=red, urlcolor=blue ]{hyperref}
\usepackage{epsfig,newlfont,amssymb,amsfonts,amsmath,bm,subfigure,palatino,mathtools,amsthm,braket,soul,enumitem,color,times,comment,geometry,float,array,physics}
\usepackage[table]{xcolor}
\usepackage[normalem]{ulem}
\newtheorem{theorem}{Theorem}

\begin{document}
\title{Floquet driven long-range interactions induce super-extensive scaling in quantum batteries}

\author{Stavya Puri$^{1,2}$, Tanoy Kanti Konar$^{2}$, Leela Ganesh Chandra Lakkaraju$^{2,3,4}$, Aditi Sen(De)$^{2}$}
\affiliation{$^1$ Department of Physics, Birla Institute of Technology and Science - Pilani, Rajasthan 333031, India}
\affiliation{$^2$ Harish-Chandra Research Institute, A CI of Homi Bhabha National Institute, Chhatnag Road, Jhunsi, Allahabad - 211019, India}
\affiliation{$^3$ Pitaevskii BEC Center and Department of Physics, University of Trento, Via Sommarive 14, I-38123 Trento, Italy }
\affiliation{$^4$ INFN-TIFPA, Trento Institute for Fundamental Physics and Applications, Trento, Italy }

\begin{abstract}
    Achieving quantum advantage in energy storage and power extraction is a primary objective in the design of quantum-based batteries. We explore  how long-range (LR) interactions in conjunction with Floquet driving can improve the performance of quantum batteries, particularly when the battery is  initialized in a fully polarized state. In particular, we analytically prove that the upper bound of the instantaneous power obtained through this system-charger duo scales quadratically with moderate system-size.
    By optimizing the driving frequency, we demonstrate that the maximum average power which is a lower bound of the instantaneous power  can achieve the super-extensive scaling with  system-size, thereby providing {\it genuine quantum advantage}.  Further,  we illustrate that the inclusion of either two-body or many-body interaction terms in the LR charging Hamiltonian leads to a scaling benefit.   We also discover that a super-linear scaling in power results from increasing the strength of interaction  compared to the transverse magnetic field  and the range of interaction with low fall-off rate, highlighting the advantageous role of long-range interactions in optimizing quantum battery charging.

\end{abstract}

\maketitle

\section{Introduction}
The miniaturization of classical devices has paved the way for the development of quantum thermal devices, including quantum batteries \cite{Alicki,Batteryreview, battery_rmp_review}, refrigerators \cite{popescu10}, transistors \cite{joulain16}, and heat engines \cite{bender2000,quan2007}. These devices leverage quantum-mechanical principles to outperform their classical counterparts while also contributing to advances in thermodynamic concepts at microscopic and nanoscale levels \cite{gemmer2004}. Having focused on quantum batteries (QBs), a crucial aspect to achieve quantum advantage is the implementation of \textit{global operations} during the charging process. Such operations lead to a super-extensive scaling of power with the increase of system-size, wherein power exhibits a nonlinear dependence on system-size \cite{dario21, battery_rmp_review,sarkar2025}, a phenomenon referred to as \textit{genuine quantum advantage} \cite{andolina2024genuinequantumadvantagenonlinear}.
In addition, several theoretical approaches have been developed to improve the performance of QBs~\cite{andolina2017,andolina2019,santos2019,andolina2020, Bera2020QB,Modispinchain,srijon2020, srijon21, srijon2021, ksen_battery_1,alba_1_20,alba_1_22,konar_battery_1,konar_battery_2,konar_battery_3,ksen_battery_4,arjmandi_pra_2022,santos_pra_2023,ksen_battery_5,Chaki2024Sep,AI_quantum_battery,sashi_pra_1,mesure_battery,remote_charging_battery,topological_quantumbattery,NV_battery,sashi_2,perciavalle2024, battery_ico,rossini_prb_2019,arjmandi_pre_2023,Gyhm2024Jan}. 
These theoretical insights have also been experimentally investigated in various physical platforms, including quantum dots \cite{wennigerexpqdots}, transmons \cite{superconducting_battrey_1, superconductQBexp, GemmeexpIBMsupercond}, organic semiconductors \cite{Quach2022Jan}, and nuclear magnetic resonance \cite{MaheshexpNMR}.

In recent years,  Floquet or time-periodic driving, also referred to as {\it Floquet engineering}, has emerged as a vital tool for exploring unique characteristics in many-body systems that are typically inaccessible in equilibrium conditions (for details, see reviews \cite{Bukov2015_review, Eckardt2015_review,Mori2023_review}). Notable examples include topological order \cite{Bukov2015_review,milad_prb_2022}, prethermalization region \cite{Eckardt2015_review,bukov_prethermalization_2023}, dynamical localization and stabilization \cite{dAlessio2016} and the creation of artificial magnetic fields \cite{BlochRMP08,BlochFloquet}. 
Furthermore,  periodic driving can be easily implemented via oscillating electromagnetic radiation in experiments with cold atoms in optical lattices \cite{Eckardt2015_review, Eckardt2017RMP} and solid state materials.
Hence, it is natural to apply Floquet evolution towards building quantum technologies~ \cite{floquet_communication_prr,floquet_quantum_computing,floquet_refrigerator,floquet_transistor,floquet_heatengine_1,floquet_heatengine_2,floquet_battery}.  In the case of QB, it was shown that although the effective periodic charging involves collective operations, it does not lead to super-extensive scaling of power \cite{floquet_battery}.

In this work,  we exhibit that long-range (LR) interactions combined with Floquet driving in the charging process can enhance the performance of the quantum batteries, resulting in quantum advantage. Specifically, we prove analytically that starting with  the initial product eigenstate of a non-interacting battery Hamiltonian, the instantaneous power of the QB with the help of a charger having long-range interactions of two bodies, commonly known as the Lipkin Meshkov Glick (LMG) model \cite{LMGmodel}, indeed can scale super extensively with system-size. Moreover, we demonstrate that such a super-linear scaling persists when charging with Floquet driving involves two types of LR \(XY\) spin models having both two- and many-body interactions, following a power-law decay. Therefore, we establish that a {\it genuine quantum advantage} \cite{andolina2024genuinequantumadvantagenonlinear} can be accomplished for both kinds of LR interactions when the coordination number is close to its maximum, the fall-off rate of interactions is low enough to provide truly long-range and the magnetic field strength is calibrated appropriately.
Importantly, such LR interactions arise naturally in experiments with trapped ions and cold-atoms on which 
Floquet driving can also be implemented 
\cite{Eckardt2017RMP,iontrapp_floquet1, iontrapp_floquet2, iontrapp_floquet3}.  
Additionally, it is well known that LR spin models are useful for quantum sensing and computation \cite{LewensteinLR, GoroshkovLR, ShajiLR, GaneshLRsensing, DebkantaLR}; our findings offer yet another advantageous use for these systems. Specifically, long-range interacting systems under Floquet interaction have been studied in topological ladders \cite{floquet_long_range_1} and photonic-lattice systems \cite{floquet_long_range_2}. 

The paper is organized in the following manner. In Sec. \ref{sec:numericalsuperlinear}, we prove analytically the scaling of the upper bound of the instantaneous power which indicates that power can scale superextensively with system-size. We then confirm this result by considering LMG model in Sec. \ref{sec:LMG}. Before concluding remarks in Sec. \ref{sec:conclu}, we investigate the scaling of power when the long-range interactions in the charging Hamiltonian follow the power law decay in Sec. \ref{sec:powerlaw}.

\section{Indication of super-linear scaling with long-range system via Floquet driving}
\label{sec:numericalsuperlinear}

To establish the quantum advantage, we investigate  the scaling of the instantaneous power at each stroboscopic time \(nT\), defined
as \(P_{ins}(t= nT) = \eval{\frac{d}{dt} (\text{Tr}[H_B \rho(t)])}_{t=nT}\), where \(H_B\) is the battery Hamiltonian and \(\rho(t)\) is the time evolved state. 
\begin{theorem}
A non-interacting \(N\)-site battery (with moderate \(N\)), when charged using a time-periodic  driving Hamiltonian with LR interactions, exhibits at most  super-extensive scaling of \(N^2\) in  instantaneous power with the increase of \(N\), i.e., 
\begin{equation}
        |P_{ins}(nT)|\leq a N^{\eta} + bN + c,  \quad\text{where} \; \eta\leq2.
\end{equation}
\end{theorem}

\begin{proof}
Let us take the battery Hamiltonian as \(H_B = h_z\sum_{j} \sigma_j^z\), while the charging Hamiltonian is
\begin{equation}
    H_{ch}(t) = H_B + H_{int}(t),
\end{equation}
where \(H_{int}\) represents the LR \(XY\) model with an open boundary condition given by,
\begin{equation}
    H_{int}(t) = \sum_{\substack{i < j \\ |i - j| \leq Z}}^{N-Z} \frac{J(t)}{\mathcal{N}|i-j|^{\alpha}} \left(\sigma_{i}^{x} \sigma_{j}^{x} + \gamma \sigma_{i}^{y} \sigma_{j}^{y} \right).
    \label{eq:H_lmg}
\end{equation}
Here, \(J(t)\) is the time-dependent interaction strength between spins at sites \(i\) and \(j\), with \(\mathcal{N}= \sum_{r=1}^{Z} \frac{1}{r^\alpha}\) (where \(r = |i-j|\)) being the Kac normalization factor \cite{Kac_jmp_1963}. The parameter \(\gamma\) and \(Z\)  represent the anisotropy factor and the coordination number (i.e., the maximum interaction range) respectively (For further details, see Appendix \ref{sec:qBatteryset}). For \(\alpha=0\), where all spin-spin interaction strengths are equal, the model reduces to the LMG model, exhibiting several exotic properties~\cite{vidal_lmg_1, vidal_lmg_2, vidal_lmg_3}. The interaction Hamiltonian in this case simplifies to
\begin{eqnarray}
     H_{int}(t) &=\nonumber \frac{J(t)}{2N} \big[(1+\gamma)(S_{+}S_{-} + S_{-}S_{+} - N) \\
     &\quad + (1-\gamma)(S_{+}^2 + S_{-}^2) \big],
     \label{eq:H_int}
 \end{eqnarray}
where \(S_l = \frac{1}{2} \sum_j \sigma_j^l\) for \(l \in \{x, y, z\}\), and \(S_{\pm} = S_x \pm iS_y\). \(J(t)\) follows a square wave modulation with period \(T = \frac{2\pi}{\omega}\), where \(\omega\) is the square wave frequency given by
\begin{equation}
 J(t) = \begin{cases} 
      +J, & nT < t < (n+1/2)T, \\
      -J, & (n+1/2)T < t < (n+1)T.
   \end{cases}
\label{eq:sqwmod_NNN}
\end{equation}

Since the drive is periodic and evaluated at stroboscopic times, the unitary evolution is given by \(U_1 = \exp[-i (H_B + H_{int}) T/2]\) and \(U_2 = \exp[-i (H_B - H_{int}) T/2]\) and the Floquet unitary after \(n\) stroboscopic periods is \(U^F(nT) = (U_2 U_1)^n\), where the updated state at stroboscopic time is given as \(\rho(t=nT) = U^F(nT) \rho(0) U^{F \dagger}(nT)\), with \(\rho(0)\) being the initial state of \(H_B\). The instantaneous power at stroboscopic times is defined as \(P_{ins}(t=nT) = \eval{\frac{d}{dt} \text{Tr}[H_B \rho(t)]}_{t=nT}\). This allows us to bound the average power \(\langle P(nT) \rangle = \frac{W(nT)}{nT}\), where we observe that \(|\langle P(nT) \rangle| \leq \max_n |P_{ins}(t=nT)| \leq \| [H_F, H_B] \|\) \(\forall n\) with \(H_F\) being Floquet Hamiltonian used to determine an upper bound on the maximum average power (see Appendix \ref{sec:analyticalsuperlinear}).

In order to obtain the Floquet Hamiltonian, we apply the Floquet-Magnus expansion (FME) \cite{magnus1954solution} (see Appendix \ref{sec:advNNandNNN} for a detailed calculation with Hamiltonian having nearest-neighbor and next nearest-neighbor interaction) and the \(k^\text{th}\) term in this expansion is given by 
\begin{eqnarray}
H_{ch}^k &\leq \frac{1}{2^{k}T} \int_{0}^{T} dt_1 \int_{0}^{t_1} dt_2 \cdots \int_{0}^{t_{k-1}} dt_k \times \nonumber\\
    &\quad [H_B \pm H_{int}, [H_B \pm H_{int}, \dots, k\text{-terms}]].
\label{eq:H_F_k_init}    
\end{eqnarray}
Since these commutators are time-independent, they simplify to \(2^{k-1}\) nested commutators of the form \(\pm2[H_{i_1}, [H_{i_2}, [H_{i_3}, \dots, [H_{i_{k-1}}, [H_B, H_{int}]]\dots]]].\)
Applying the triangle inequality, we obtain
\begin{align}
 \| H_{ch}^k \| &\leq \frac{T^{k-1}}{2^{k-1} k!} \sum_{i_1} \sum_{i_2} \dots \sum_{i_{k-1}} \nonumber\\
 &\quad \Big[ \| [H_{i_1}, [H_{i_2}, \dots, [H_{i_{k-1}}, [H_B, H_{int}]]\dots]] \| \Big].
\label{eq:H_F_k_final}
\end{align}
By considering the norm of each \(k^\text{th}\) term in the effective Hamiltonian, we arrive at the upper bound on the instantaneous power \cite{dario21} at stroboscopic times as:
\begin{align}
    |P_{ins}(nT)| \nonumber &\leq N\left[\sum_{k=1}^{N-1}\frac{T^{k-1}}{k!}\frac{(k+3)}{2}\right] \\& \nonumber
    + N\left[\sum_{k=0}^\infty \frac{T^{k+N-1}}{(k+N)!} \frac{\sum_{i=0}^{N-2}(i+2)\binom{k+N-1}{i}}{2^{k+N-1}}\right] \\& 
    + N^2\left[\sum_{k=0}^\infty \frac{T^{k+N-1}}{(k+N)!}\frac{\sum_{i=N-1}^{k+N-1} \binom{k+N-1}{i}}{2^{k+N-1}} \right]
\label{eq:lmg_bound}
\end{align}
where we assume \(h_z=1\) and \(J=1\), though the result remains valid for other parameter choices. Since the factor inside the third bracket does not scale with \(N^{-1}\) (especially when N is moderate), we obtain that the right hand side scales as \(N^{2}\) and hence the proof (see Appendix \ref{sec:analyticalsuperlinear} for further details).
\end{proof}


Analytical and numerical simulations suggest that this super-extensive scaling cannot be observed when \(N\) is very large. Such a behavior can be explained from the second and third terms in Eq. (\ref{eq:lmg_bound}) which decreases with the increase of  \(N\) as it contains \(((k+N)! \; 2^{k+N-1})\) in the denominator resulting in linear scaling.

\begin{figure}    \includegraphics[width=\linewidth]{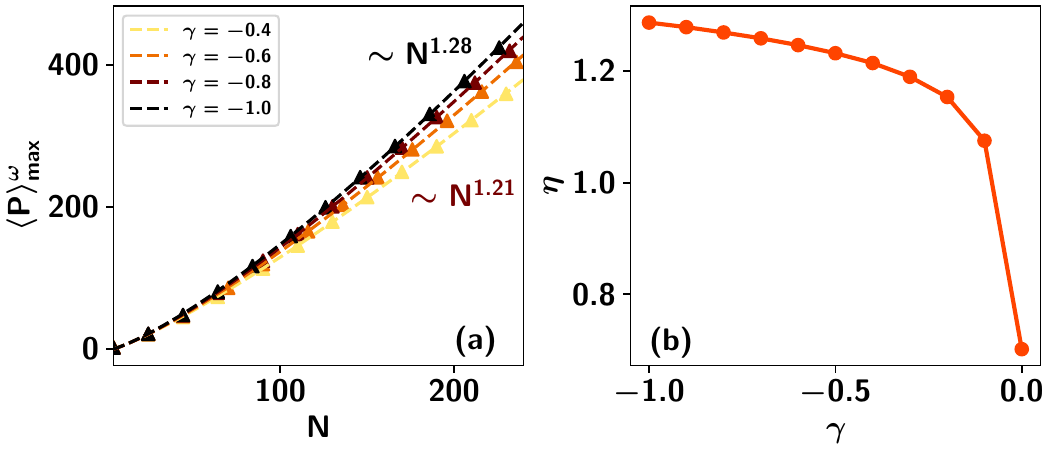}
    \caption{{\bf Scaling analysis for LMG charging Hamiltonian.} (a) \(\langle P \rangle _{\max}^\omega\) (ordinate) with system-size \(N\) (horizontal axis), when the charging Hamiltonian is the two-body LR \(XY\) model with \(\alpha =0\). By least-square fitting method, we find that it scales as \(\langle P \rangle _{\max}^\omega \sim a N^{\eta} + b\), where \(a\) and \(b\) are constants, with \(\eta>1\) as mentioned in the figure.   (b) Dependence of \(\eta\) (ordinate) on \(\gamma\) (abscissa) which shows the scaling upto \(\approx 1.3\). Other parameters are \(h_z=1\), and \(J=20\). All the axes are dimensionless. }
    \label{fig:LMG_MaxPvN_GammaDep}
\end{figure}



\begin{figure}
    \includegraphics[width=\linewidth]{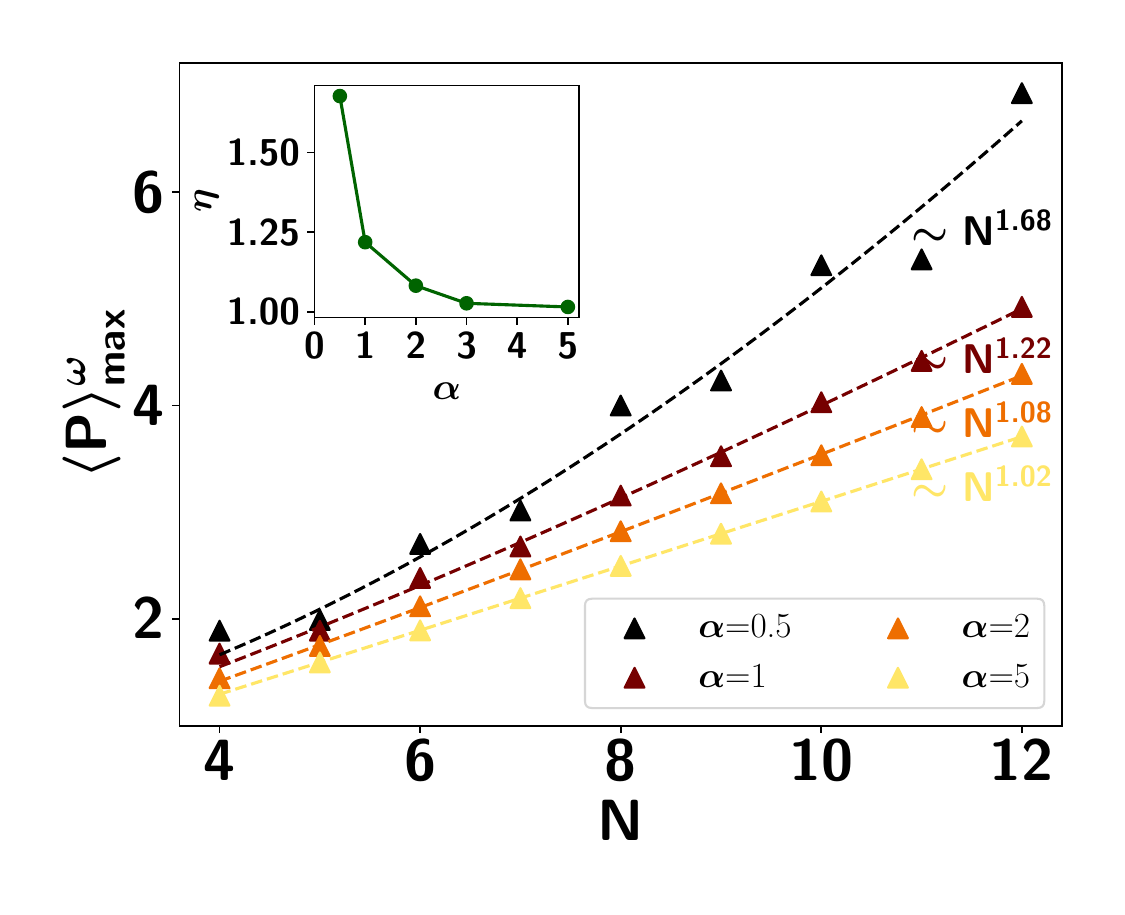}
    \caption{{\bf Scaling of LR \(XY\) charging Hamiltonian.} Behavior of \(\langle P \rangle _{\max}^{\omega}\) against system-size \(N\) with \(h_z=0.5\), \(J = 5\), \(\gamma=-1\). Periodic driving along with LR interaction induces super-extensive scaling, i.e., when \(\alpha <2\),  \(\langle P \rangle _{\max}^{\omega} \sim a N^{\eta} +b \)  with \(\eta>1\) and \(a\), \(b\) being constants. Inset depicts how super-linear scaling exponent \(\eta\) (vertical axis) changes with LR interaction strength \(\alpha\) (horizontal axis).  Both the axes are dimensionless.   }
    \label{fig:lr_MaxPvN_scaling}
\end{figure}

\subsection{Scaling of power in the LMG model.}
\label{sec:LMG}
Let us now exhibit that the scaling of maximum average power which is defined as \(\langle P \rangle_{\max}^{\omega}=\max_{{n,\omega}}{\frac{W(nT)}{nT}}\) (where maximization is performed over stroboscopic time \(t\) and frequency range \(\omega\)~\cite{Omegamaxreason})  can saturate the bound obtained in Theorem 1 although \(|\langle P(nT) \rangle| \le \max_n|P_{ins}(nT)|\) (see Appendix \ref{sec:analyticalsuperlinear}). Specifically, we observe that by taming the parameter suitably, \(\langle P \rangle _{\max}^{\omega} \sim N^{\eta}\) with \(1< \eta \le 1.5\) when \(-1 \le \gamma \le 0\) (see Fig. \ref{fig:LMG_MaxPvN_GammaDep}). Eg. when \(h_z/J = 0.05\) and \(\gamma = -1.0\), the scaling equation reads as \(\langle P \rangle _{\max}^\omega = a N^{1.28} + b\) where \(a= 0.39\), \(b= -2.24\) and  mean square error is \( 0.16 \% \) (see Fig. \ref{fig:LMG_MaxPvN_GammaDep}).  It also highlights that \(\gamma =-1\) is the best choice to obtain super-linear scaling compared to any values from \(\gamma <0 \). This can be intuitively understood as the norm of the commutator maximizes when $\gamma < 0$. Moreover, we find that the scaling exponent, \(\eta\), depends on the ratio between the strength of the magnetic field and  the interaction, \(h_z/J\). Precisely, when the interaction strength dominates over the strength of the magnetic field, i.e.,  when \(h_z/J\ll 1\) (in the paramagnetic phase), the non-linear scaling  with \(\eta>1\) is observed while  for \(h_z/J>1\) (in the ferromagnetic phase), the scaling remains  extensive although the magnitude of \(\langle P \rangle _{\max}^{\omega}\) is higher than  that for \(h_z/J>1\). Notice, further, that with the increase of \(J\) such that \(h_z/J \ll 1\), the scaling exponent increases and, finally, it saturates to \(N^{1.5}\). This behavior can be explained from the effective \(k\)-body interactions in the Floquet Hamiltonian as in Eq. (\ref{eq:H_F_k_init}) --  the terms having greater k-body interactions have more number of \(H_{int}\) in the commutator which results in more contributions in power by the strength of interaction than the magnetic field, i.e., in the paramagnetic phase. (see mathematical description in Appendix \ref{sec:analyticalsuperlinear}).

\begin{figure}    \includegraphics[width=\linewidth]{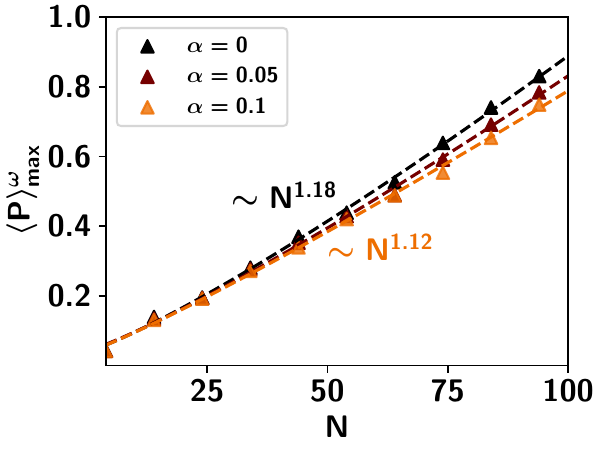}
    \caption{{\bf Scaling analysis for LR extended \(XY\) charging Hamiltonian.} \(\langle P \rangle _{\max}^\omega\) (ordinate) with system-size \(N\) (horizontal axis) for different fall-off rates \(\alpha\). Other parameters are \(h_z=0.02\), \(J=1\) and \(\gamma=-1\). Again, by least square fitting, we find \(\langle P \rangle _{\max}^\omega \sim a N^\eta + b\) where \(a\) and \(b\) are constants and \(\eta >1\) as mentioned in the figure for low \(\alpha\) values.  Both the axes are dimensionless. }
    \label{fig:extended_pmaxom}
\end{figure}

\section{Super-linear scaling of power in long-range Hamiltonian with power-law decay}
\label{sec:powerlaw}

So far, we have analyzed the scaling behavior of power for the LMG model, where all interactions are of the same order although, for physical systems like trapped ions \cite{Eckardt2017RMP,iontrapp_floquet1, iontrapp_floquet2, iontrapp_floquet3}, a more natural choice can be varied interaction strengths. To address this query, we consider two systems modeled by  Eq. (\ref{eq:H_int}) and  the extended \(XY\) model \cite{GoroshkovLR,VodolaLR,DebasisJacekLR,GaneshDebasisLR}, given by
 \begin{eqnarray}
    H_{int}^{extXY}=\sum_{j=1}^{N} \sum_{r=1}^{\frac{N}{2}}\frac{J(t)}{4\mathcal{N}r^\alpha} (\sigma_j^x\mathbb{Z}_r^z\sigma_{j+r}^x+\gamma \sigma_j^y\mathbb{Z}_r^z\sigma_{j+r}^y),
   \label{eq:H_extXY}
\end{eqnarray}
where $\mathbb{Z}_r^z = \prod_{l=j+1}^{j+r-1}\sigma_l^z$, with $\mathbb{Z}_1^z=\mathbb{I}$, with periodic boundary condition. Since this model can be solved analytically \cite{lieb1961, barouch_pra_1970_1, barouch_pra_1970_2}, \(\langle P \rangle_{\max}^{\omega}\) can be computed for large system sizes. We now examine how the scaling and other characteristics of \(\langle P \rangle _{\max}^{\omega}\) get affected as \(\alpha\) increases, i.e., as the strength of long-range interactions decreases. 
 When \(\alpha \le 2\) and \(\gamma=-1\), numerical simulations reveal that the charging Hamiltonian in Eqs. (\ref{eq:H_int}) and (\ref{eq:H_extXY}) continue to exhibit super-extensive power scaling in the paramagnetic phase (\(h_z/J \ll 1\)).
 
 \textit{Scaling exponent for LR \(XY\) model} The scaling exponent \(\eta\) increases as \(\alpha\) decreases, reaching its maximum as \(\alpha \to 0\) (as illustrated in Fig. \ref{fig:lr_MaxPvN_scaling}). This clearly demonstrates that long-range interactions, both in the truly long-range (LR) and quasi-long-range (quasi-LR) regimes, combined with periodic driving, lead to a significant enhancement in the scaling exponent of power (ranging from \(\sim 1.68\) for \(\alpha =0.5\) to \(\sim 1.08\) with \(\alpha =2\)) compared to the short-range interactions in the charging Hamiltonian. This suggests that the observed super-linear scaling is possibly linked to the ability of the driving Hamiltonian to efficiently distribute entanglement across sites \cite{ganeshPLALR}. Furthermore, when \(\alpha \gg 2\), even in the presence of Floquet driving, the power only has a linear dependence on \(N\). For example, at \(\alpha = 5\) and \(h_z/J = 0.1\), the maximum power follows the scaling relation, \(\langle P \rangle _{\max}^{\omega} \sim a N + b,\) with \(a = 0.29\) and \(b = 0.15\). \\
 \textit{Work and power of extended \(XY\) model.} When the initial state is chosen to be \(\ket{0_k}\) in the momentum space, i.e., the ground state of \(H_B\), the work output can be obtained analytically as  \(W(nT)=\sum_{k=1}^{N/2}2h_z(1-(n_k^z)^2)\sin^2(n\cos^{-1}u_k^0)\) at \(n^{th}\) stroboscopic time, where \(n_k^z\) and \(u_k^0\) are complex functions of \(\alpha, T\) and parameters of the extended \(XY\) Hamiltonian (see Appendix \ref{sec:extendedXY} for further details). It clearly indicates that \(W(nT)\) depends non-linearly on the frequency of the Floquet drive and by calibrating \(\omega, \alpha \text{ and } h_z/J\), we again obtain \(\langle P \rangle _{\max}^{\omega} \sim aN^\eta +b\) with \(\eta>1\) (see Fig. \ref{fig:extended_pmaxom}). Unlike the LR \(XY\) model, non-linear dependence of power on \(N\) cannot be found in the quasi-LR regimes of this model.

\begin{figure}\includegraphics[width=0.9\linewidth]{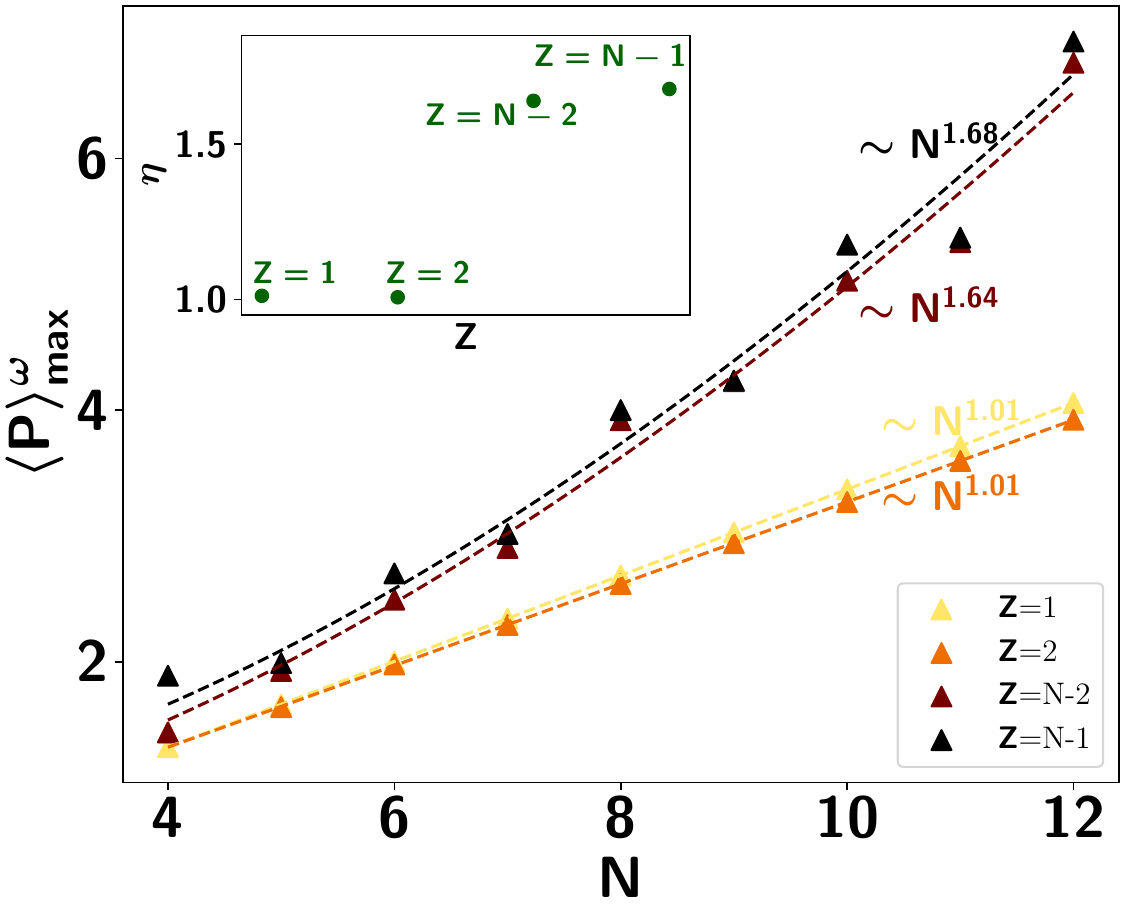}
    \caption{\textbf{Scaling of $\langle P \rangle _{\max}^{\omega}$ (ordinate) vs system-size \(N\) (abscissa) with \(\gamma = -1\) and \(\alpha=0.5\).} Dependency upon \(N\) is displayed for different coordination number \(Z\) in the long-range \(XY\) model. Other system parameters are \(J = 5\) and \(h_z = 0.5\). Inset displays \(\eta\) (vertical axis) with respect to \(Z\) (horizontal axis). It clearly indicates the influence of \(Z\) on achieving super-extensive scaling in \(\langle P \rangle _{\max}^{\omega}\) as \(\eta\) increases with \(Z\).  All the axes are dimensionless.}
    \label{fig:lr_PvN_Zvar}
\end{figure}

\subsection{ Response of coordination number in the scaling}
Let us now ask the following question --  {\it ``how does this super-linear scaling depend on the range of interactions in the long-range domain?''} The scaling of \(\langle P \rangle _{\max}^{\omega}\) changes drastically with the co-ordination number \(Z\) which as well depends upon the long-range inetraction strength, \(\alpha<1\) and \(J\gg h_z\). For a fixed \(\alpha\), the scaling   significantly gets improved with increment of \(Z\) (see Fig. \ref{fig:lr_PvN_Zvar}). This is possibly due to the fact that  strong long-range interaction is capable of creating high entanglement among all the sites, leading to a greater quantum advantage in QB in comparison with the nearest-neighbor or few neighbor interacting charger. 
The contrasting behavior is evident from investigating two extreme cases: \(Z=2\) gives linear scaling, i.e., \(\eta =1\) while \(Z = N-1\), \(\eta =1.68\) with \(\alpha \leq 2\) (see Fig. \ref{fig:lr_PvN_Zvar})
that with the increase of \(Z\), the scaling exponent, \(\eta\),  increases. 
Moreover, the increment of \(\langle P \rangle _{\max}^{\omega}\) with \(Z\) also occurs due to the contribution of higher order \(k\)-body terms in the computation of power as in Eq. (\ref{eq:lmg_bound}).

\section{Conclusion}  
\label{sec:conclu}

In this study, we explored the role of Floquet evolution in enhancing the charging process of quantum batteries (QB), focusing on the instantaneous and maximum average power. We demonstrated that by appropriately tuning the Floquet frequency, quantum advantage can be achieved when the charging is governed by a long-range interaction Hamiltonian. Specifically, when charging is governed by a long-range interacting model, a clear quantum advantage emerges in the deep long-range regime which we established by analytically proving the upper bound on the instantaneous power and by analyzing the the maximum average power. In this regime, the power scales  super-linearly with system-size. In particular, the scaling behavior depends on both the parameters of the charging Hamiltonian and the strength of the long-range interactions, characterized by the fall-off rate.  We find that as the fall-off rate  increases, the scaling decreases, suggesting that weaker long-range interactions diminish the quantum advantage. Our analysis shows that  the impact of fall-off rate and the range of interactions in the scaling analysis of power is not at all straightforward. Specifically, the effect of increasing long-range interactions on the charging  can be either beneficial or detrimental, depending on the specific region of  the parameters under consideration. This non-monotonic behavior reflects the intricate interplay between the interaction strength, the Floquet frequency, and the underlying Hamiltonian structure.

The results also reveal the potential of periodic driving and long-range interactions as crucial tools in the development of effective quantum energy storage systems and quantum technologies, in general, emphasizing the delicate balance that must be found between system parameters to attain maximum performance.

\acknowledgments
 
 We  acknowledge the use of \href{https://github.com/titaschanda/QIClib}{QIClib} -- a modern C++ library for general purpose quantum information processing and quantum computing (\url{https://titaschanda.github.io/QIClib}), and the cluster computing facility at the Harish-Chandra Research Institute. This research was supported in part by the INFOSYS scholarship for senior students. LGCL received funds from project DYNAMITE QUANTERA2-00056 funded by the Ministry of University and Research through the ERANET COFUND QuantERA II – 2021 call and co-funded by the European Union (H2020, GA No 101017733).  Funded by the European Union. Views and opinions expressed are however those of the author(s) only and do not necessarily reflect those of the European Union or the European Commission. Neither the European Union nor the granting authority can be held responsible for them. This work was supported by the Provincia Autonoma di Trento, and Q@TN, the joint lab between University of Trento, FBK—Fondazione Bruno Kessler, INFN—National Institute for Nuclear Physics, and CNR—National Research Council.

\appendix

\section{Detailed description of quantum battery and the performance quantifiers}
\label{sec:qBatteryset}



{\it Quantum battery.} We prepare the initial state of the quantum battery as a thermal equilibrium state, \(\rho(0) = \frac{\exp(-\beta H_B)}{\text{Tr}(\exp(-\beta H_B))}\) where \(\beta = 1/k_B T\) with \(k_B\) being the Boltzmann constant and \(T\) representing the temperature.  We choose the battery Hamiltonian to be
  \(  H_B = h_z\sum_j \sigma^z_j\),
where \(h_z\) is the strength of the local magnetic field, quantifying the local energy gap of each subsystem and \(\sigma^z\) is the \(z\)-component of the Pauli matrix.  Note that when \(\beta\to\infty\),  \(\rho(0)\) reaches to the ground state of  \(H_B\), i.e.,  \(\ket{\psi(t=0)}=\ket{00\ldots 0}\) with \(\ket{0}\) being the ground state of \(\sigma^z\).

{\it Charging the battery.} In order to obtain super-linear scaling of QB,  we will demonstrate that the charging operation plays an important role.  We incorporate   {\it two} important components in the charging  Hamiltonian of the QB, \(H_{ch}(t)=H_B+H_{int}(t)\), which are different from the previous protocols known in the literature \cite{bera2019,floquet_battery}. These two crucial ingredients are as follows: \\
(1a) We choose variable-range interacting Hamiltonian  as  \(H_{int}\), given by 
\begin{eqnarray}
     H_{int}^{LR} = \sum_{\substack{i < j \\ |i - j| \leq Z}}^{N-Z}{\frac{J(t)}{\mathcal{N}|i-j|^{\alpha}}(\sigma_{i}^{x}\sigma_{j}^{x}} + \gamma\sigma_{i}^{y}\sigma_{j}^{y}),
     \label{eq:H_int}
\end{eqnarray}
which is responsible to build  a multi-site correlation between different subsystems of the QB. Here \(J(t)\) is the time-dependent interaction strength between the spins at site, \(i\) and  \(j\) with \(\mathcal{N}= \sum_{r=1}^{Z} \frac{1}{r^\alpha}, \, r = |i-j|\), known as the Kac normalization factor \cite{Kac_jmp_1963},  \(\gamma\) and \(Z\) represent the anisotropic factor and the coordination number, i.e., the distance between sites, \(i\) and \(j\) respectively and \(\sigma^k\) (\(k=x,y,z\)) are the Pauli matrices. We also assume a power-law functional form for the decay of the interactions with the increasing distance between the spins such that \(\alpha\) corresponds to the fall-off rate of this power law-decay.  In this case, an open-boundary is considered.  By changing  \(\alpha\) values, we can have  long-range interactions  with  \(0 \leq \alpha \leq 1\), quasi LR interactions for  \(1 < \alpha \leq 2\) and short-range interaction when \(\alpha >2\).  Note that \(\alpha =0\) corresponds to the LMG model \cite{LMGmodel}. In this work,  we study the gain in QB by varying both \(\alpha\) and \(Z\). \\
(1b) Another model that we choose for charging is the  extended \(XY\) model 
which can be solved analytically by Jordan-Wigner transformation \cite{lieb1961, barouch_pra_1970_1, barouch_pra_1970_2}. In this case,  the interacting Hamiltonian reads as
\begin{equation}
    H_{int}^{exXY}=\sum_{j=1}^{N} \sum_{r=1}^{\frac{N}{2}}\frac{J(t)}{4\mathcal{N}r^\alpha} (\sigma_j^x\mathbb{Z}_r^z\sigma_{j+r}^x+\gamma \sigma_j^y\mathbb{Z}_r^z\sigma_{j+r}^y),
    \label{eq:LRextXY}
\end{equation}
where $\mathbb{Z}_r^z = \prod_{l=j+1}^{j+r-1}\sigma_l^z$, with $\mathbb{Z}_1^z=\mathbb{I}$, with \(\alpha\) being the strength of power-law decay 
and the Kac-scaling factor respectively as given in Eq. (\ref{eq:H_int}).  Here a periodic boundary condition is considered.
We are interested to find out whether the extended \(XY\) model in Eq. (\ref{eq:LRextXY})  involving both \(N\)-body interactions and long-range interactions can provide similar or any additional  benefit compared to the
long-range models in Eq. (\ref{eq:H_int}).

 (2) We consider the evolution of QB through square wave with time period, \(T=\frac{2\pi}{\omega}\), given as
\begin{equation}
 J(t) = \begin{cases} 
      +J ; & nT < t < (n+1/2)T \\
      -J ; & (n+1/2)T < t < nT,  \\ 
   \end{cases}
   \label{eq:sqwmod_NNN}
\end{equation}
where \(\omega\) represents the frequency of the periodic driving. Given the square wave form of the periodic drive performed at stroboscopic times, we use the unitary of the form
\begin{gather}
    H_1 = H_B+H_{int}, \, H_2 = H_B-H_{int},\nonumber\\
    U_1 = \exp[-i (H_B+H_{int}) \frac{T}{2}], \,
    U_2 = \exp[-i (H_B-H_{int}) \frac{T}{2}],\nonumber\\
    U^F(T) = U_2U_1,
    \label{eq:U_num}
\end{gather}
where \(U^F\) is unitary corresponding to the periodic time, \(T\). The \(n^{th}\) stroboscopic evolution from the initial state, \(\ket{\psi(t=0)}\), is given as \(\ket{\psi(t=nT)}=(U^F)^n\ket{\psi(t=0)}\).

{\it Performance quantifier.} In order to certify the performance of the battery with respect to its capability in storing and extractable energy, we compute the work stored in a given battery at each stroboscopic time, \(nT\), as \(W(t=nT) = \Tr(H_B\rho(t)) - \Tr(H_B \rho(0))\), where \(\rho(0)\)  and \(\rho(t)\) are  the initial and  the evolved state of the battery. Since changing certain parameters of the Hamiltonian can cause extraction of more power, making the design unreasonable,  we normalize the Hamiltonian. It makes its spectrum to be bounded by \([~-1, 1]\), irrespective of any system parameters which is given as
 \( (E_{\max}-E_{\min})^{-1}[2H_B-(E_{\max}+E_{\min})\mathbb{I}]\rightarrow H_B\).

\begin{figure}    \includegraphics[width=\linewidth]{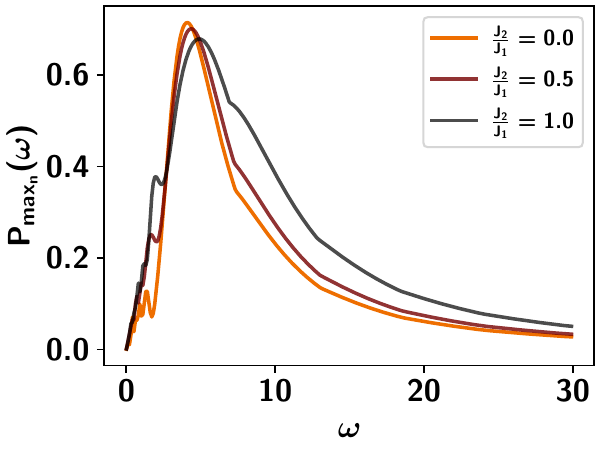}
    \caption{Behavior of \(P_{\max_n} (\omega)\) (ordinate)  against \(\omega\) (abscissa) of the \(XY\) model with the increase of  \(J_2/J_1\) (from dark to light lines).  Here \(h_z = 0.5\),    \(\gamma = -1\) and \(N = 8\). The initial state is the product ground state of the battery Hamiltonian, \(H_B= h_z \sum_j \sigma^z_j\) (see Eq. (\ref{eq:H_J1_J2})).  We observe that there exists \(\omega\) values for which \(P_{\max_n} (\omega)\) attain its maximum.   All the axes are dimensionless.}
    \label{fig:NNN_MaxPvsomega_merged}
\end{figure}

We are interested to investigate two quantities to assess the performance of the battery.\\
\textbf{(1) Instantaneous power. } The instantaneous power of the battery is defined as
\begin{equation}
     P_{ins}(t) = \frac{d}{dt}\text{Tr}(H_B\rho(t)-H_B\rho(0))
\end{equation}
which reduces to \(P_{ins}(t)=\frac{d}{dt}\text{Tr}(H_B\rho(t))\) as the initial state and the Hamiltonian is time independent. Using the Liouville von-Neumann equation at stroboscopic times \(\frac{d}{dt}\rho(t)|_{t=nT} = -i[H_{ch}^F,\rho(nT)]\), we obtain \(|P_{ins}(t)|_{t=nT} = |\text{Tr}(H_B[H_{ch}^F,\rho(nT)])|\).\\\\
\textbf{(2) Maximum average power. } One can tune the frequency of the Floquet driving. More precisely, when \(\omega\to 0\),  the system evolves unitarily and the time period is very high for which power through stroboscopic time become very small, i.e., maximization of power over stroboscopic time, \(\langle P \rangle_{\max_{n}}(\omega)=\underset{n}{\max}\frac{W(nT)}{nT}\), \(\langle P \rangle_{\max_{n}}(\omega)\to 0\) while for \(\omega\to \infty\), QB evolves through an average Hamiltonian \(H_B\) resulting in \(\langle P \rangle_{\max_{n}} (\omega)\to 0\) (See Fig. \ref{fig:NNN_MaxPvsomega_merged}). Such a scenario encourages us to define maximum power through maximization over frequency of the Floquet driving,  which is defined as 
\begin{equation}
     \langle P \rangle _{\max}^{\omega} = \max_{\substack{n,\omega}} \frac{W(nT)}{nT},
\end{equation}
where the maximization is performed over stroboscopic time, \(t\) and the frequency range, \(\omega\). Note that the optimization over \(\omega\)  does not appear in the unitary evolution and is not  considered for  Floquet charging in literature (see Ref.  \cite{floquet_battery}).  Therefore, it highlights  that non vanishing  maximum power  depends upon \(\omega\) and genuine quantum advantage through Floquet charging can only be confirmed through optimization over frequency as well as stroboscopic time. In other words, we are interested to identify the favorable situation in which quantum advantage can be maximized as done in case of other quantum tasks  \cite{nielsenchuang2000, Giolloydmaccone06, de2011quantumadvantageqcom, Gisin2007}. Note here that the normalization of the battery Hamiltonian $H_B$ described above is not performed when we  compute the scaling behavior of \( \langle P\rangle_{\max}^{\omega}\) with the system-size \(N\) since we want to compare our results with the known results in literature computed without normalization.

\section{Super-linear scaling with long-range interacting charging Hamiltonian via Floquet driving}
\label{sec:analyticalsuperlinear}

 In order to provide the signature of super-extensive scaling of maximum average power, we now investigate how the upper bound of the intantaneous power scales with system size. In particular, the bound on \(|P_{ins}(t)|_{t=nT}\) is given as 

 \begin{widetext}
     \begin{eqnarray}
    |P_{ins}(t)|_{t=nT} &=& |\text{Tr}(H_B[H_{ch}^F,\rho(nT)])|= |\text{Tr}(\rho(nT)[H_{ch}^F,H_B])|=|\bra{\psi(nT)}[H_{ch}^F,H_B]\ket{\psi(nT)}|\nonumber \\
    &\leq& \max_{\substack{n}}|\{\bra{\psi(nT)}[H_{ch}^F,H_B]\ket{\psi(nT)}\}|\le ||[H_{ch}^F,H_B]||.
    \label{eq:P_inst_max}
\end{eqnarray}
Now the bound on average power is given as 
\begin{eqnarray*}
    |\expval{P(nT)}|&=& \frac{1}{nT}\Bigg |\int^{nT}_0 \frac{d}{dt}E_B(t) dt\Bigg |\leq \frac{1}{nT}\int^{nT}_0 \Big  |\frac{d}{dt}E_B(t) \Big | dt= \frac{1}{T}\int^{T}_0 \Big |\text{Tr}(H_B[H_{ch}^F,\rho(nT)]) \Big | dt \\&=& \leq \frac{1}{T}\int^{T}_0 \max_n \Big |P_{ins}(nT) \Big | dt \nonumber=  \max_n \Big |P_{ins}(nT) \Big|.\nonumber
\end{eqnarray*}
 \end{widetext}

Hence we utilize the bound on instantaneous power to provide the evidence of quantum advantage in this system. Also, for convenience we will henceforth denote \(P_{ins}(t)\) by \(P(t)\).

\subsection{Detailed proof of Theorem 1}
\label{subsec:analyProof}
\begin{widetext}
Let us consider the aforementioned battery Hamiltonian \(H_B\) charged using the \(H_{ch}=H_B + H_{int}^{LR}\) (Eq. (\ref{eq:H_int})).
Following the results in ref. \cite{dario21} 
\begin{eqnarray}
        |P(t)| &\leq& \sum_{k=1}^{N} k ||H_{ch}||\,||H_s-E_{s_{\min}}||\nonumber,\\
        |P(nT)| &\le& \sum_{k=1}^{N} k ||H_{ch}^F||\,||H_s-E_{s_{\min}}||,
        \label{eq:P_inst_bound}
\end{eqnarray}
where \(H_{ch}^F\) is the time-independent Floquet Hamiltonian. In our case, \(||H_s-E_{s)\min}||\) is constant. The Floquet Hamiltonian is calculated using the Floquet Magnus expansion (FME) and the \(k^{th}\) term of this FME obeys the following inequality,
   \begin{equation*}
    H_{ch}^k\leq\frac{1}{2^{k}T}\int_{0}^{T}dt_1\int_{0}^{t_1}dt_2\int_{0}^{t_2}dt_3\ldots\int_{0}^{t_{k-1}}dt_k[\pm H_{p_1},[\pm H_{p_2},[\pm H_{p_3},\ldots\, k-1\text{ terms }\ldots[\pm H_{p_k},\pm H_{p_{k+1}}]\ldots]]]
\end{equation*}
where \(H_{p_m}\,\in\,\{H_1,H_2,H_1+H_2,H_1-H_2\}\); \(m=1,2,3 \ldots k+1\) which are all time-independent. This factor of \(2^k\) appears because, as mentioned below, the commutator term can expand into \( \sim 2^k\) number of non vanishing terms. If the commutator terms are not restricted by at least the factor, it can cause the series to inherently diverge. Using Eq. (\ref{eq:U_num}), we express the commutator in terms of \(H_B \text{ and } H_{int}\) which modifies  Eq. (\ref{eq:H_F_k_og}) as
\begin{equation}
        H_{ch}^k\leq\frac{1}{2^{k}T}\int_{0}^{T}dt_1\int_{0}^{t_1}dt_2\int_{0}^{t_2}dt_3\ldots\int_{0}^{t_{k-1}}dt_k[H_B\pm H_{int},[H_B\pm H_{int},[H_B\pm H_{int},\ldots\, k-1\text{ terms }\ldots[H_B\pm H_{int}, H_B\pm H_{int}]\ldots]]]
        \label{eq:H_F_k_og}
\end{equation}
Using the fact that the commutator in the \(k^{th}\) term of FME is time-independent and the Floquet driving has square wave pattern, we can write the \(k^{th}\) term as
\begin{align*}
    I^k &= \int_{0}^{T}dt_1\int_{0}^{t_1}dt_2\int_{0}^{t_2}dt_3\int_{0}^{t_3}dt_4\,....\int_{0}^{t_{k-2}}dt_{k-1}\int_{0}^{t_{k-1}}dt_k=\sum_{l=1}^{k-1}I_l^k,
\end{align*}
where \(I_l^k\) is given as
\begin{equation*}
        I_l^k=\left(\int_{\frac{T}{2}}^{T}dt_1\int_{\frac{T}{2}}^{t_1}dt_2\int_{\frac{T}{2}}^{t_2}dt_3\int_{\frac{T}{2}}^{t_3}dt_4....\int_{\frac{T}{2}}^{t_l}dt_{l+1}\right)\left(\int_{0}^{\frac{T}{2}}dt_{l+2}\int_{0}^{\frac{T}{2}}dt_{l+3}\int_{0}^{\frac{T}{2}}dt_{l+4}....\int_{0}^{\frac{T}{2}}dt_{k}\right).
\end{equation*}
Now applying change of variables, \(u_i = t_i-\frac{T}{2}\) in the above integral, we obtain
\begin{eqnarray*}
    I_l^k&=&\left(\int_{0}^{\frac{T}{2}}du_1\int_{0}^{u_1}du_2\int_{0}^{u_2}du_3\int_{0}^{u_3}du_4....\int_{0}^{u_l}du_{l+1}\right)\left(\int_{0}^{\frac{T}{2}}dt_{l+2}\int_{0}^{\frac{T}{2}}dt_{l+3}\int_{0}^{\frac{T}{2}}dt_{l+4}....\int_{0}^{\frac{T}{2}}dt_{k}\right)\\
    &=&\left(\frac{(\frac{T}{2})^{l+1}}{(l+1)!}\right)\left(\frac{(\frac{T}{2})^{k-l-1}}{(k-l-1)!}\right)=\frac{(\frac{T}{2})^{k}}{k!}\binom{k}{l},
\end{eqnarray*}
which provides the coeffiecent of the \(k^{th}\) FME, given as 
\begin{eqnarray*}
    I^k &=& \sum_{l=1}^{k-1}I_l^k\leq \sum_{l=0}^{k}I_l = \sum_{l=0}^{k}\frac{(\frac{T}{2})^{k}}{k!}\binom{k}{l}=\frac{T^{k}}{k!}.
\end{eqnarray*}
Hence, Eq. (\ref{eq:H_F_k_og}) simplifies to
\begin{equation}
     H_{ch}^k\leq\frac{T^{k-1}}{2^{k}k!}[H_B\pm H_{int},[H_B\pm H_{int},[H_B\pm H_{int},\ldots\, k-1\text{ terms }\ldots[H_B\pm H_{int}, H_B\pm H_{int}]\ldots]]].
     \label{eq:H_F_k_init}
\end{equation}
After expanding the commutators in Eq. (\ref{eq:H_F_k_init}), we obtain \(2^{k-1}\) number of individual commutators of the form
\begin{equation}
    \pm2[H_{i_1},[H_{i_2},[H_{i_3},[H_{i_4}, \ldots,[H_{i_{k-1}},[H_B,H_{int}]]\ldots]]]]
    \label{eq:H_F_k_comm_gen}
\end{equation}
such that \(H_{i_m}\,\in\,\{H_B,H_{int}\}; \, m=1,2,3...k-1\). Referring to Eq. (\ref{eq:P_inst_bound}), we consider the operator norm \(||H_{ch}^F||\) and invoke triangle's inequality for the commutator in Eq. (\ref{eq:H_F_k_init}) expanding into \(2^{k-1}\) possible terms,
\begin{eqnarray*}
    ||[H_B\pm H_{int},[H_B\pm H_{int},[H_B\pm H_{int},\ldots\, k-1\text{ terms }\ldots[H_B\pm H_{int}, H_B\pm H_{int}]\ldots]]]||\\
    \leq \sum_{i_1}\sum_{i_2}\ldots\sum_{i_k-1}2||\,[H_{i_1},[H_{i_2},[H_{i_3},[H_{i_4}, \ldots,[H_{i_{k-1}},[H_B,H_{int}]]\ldots]]]]\,||.
\end{eqnarray*}
This gives the following bound for \(||H_{ch}^k||\):
\begin{equation}
    ||H_{ch}^k||\leq\frac{T^{k-1}}{2^{k-1}k!}\sum_{i_1}\sum_{i_2}\ldots\sum_{i_k-1}||\,[H_{i_1},[H_{i_2},[H_{i_3},[H_{i_4}, \ldots,[H_{i_{k-1}},[H_B,H_{int}]]\ldots]]]]\,||.
    \label{eq:H_F_k_final}
\end{equation}
Owing to the nature of the non-interacting battery Hamiltonian and the two-body interacting structure of charging LMG Hamiltonian, the commutator in Eq. (\ref{eq:H_F_k_comm_gen}) holds the property such that if there are '\(l\)' number of \(H_{int}\) terms out of the \(k-1\) terms in the commutator, it corresponds to \(l+2\) body interaction term. Hence, the \(k^{th}\) term of the Floquet Hamiltonian contains atmost a \((k+1)-\)body interaction. The FME contains infinitely many terms, however, this property is valid for \(k = N-1\) after which all the terms contain an \(N-\)body interaction. So accordingly, we divide the FME into 2 regimes, \(k\leq N-1\) and \(k \geq N\). For the above mentioned two regimes, under the assumption \(h_z=1\), \(J=1\) and \(\gamma=-1\),  we provide an upper bound for the instantaneous power at stroboscopic times which are as follows:

{\it (A)} \(\mathbf{k\leq N-1}\). While maintaining the extensivity of each individual commutator, the contribution of the \(k^{th}\) term of the FME in instantaneous power is given as
\begin{eqnarray}
    |P(nT)|=\sum_{k=1}^{N-1}|P(nT)^{(k)}|&\leq& \sum_{k=1}^{N-1}k||H_{ch}^k|| \nonumber\\
    &\leq&\sum_{k=1}^{N-1}\frac{T^{k-1}}{2^{k-1}k!}\left( 2\binom{k-1}{0}N+3\binom{k-1}{1}N+4\binom{k-1}{2}N+.....+(k+1)\binom{k-1}{k-1}N\right)
    \label{eq:P_dep_Hint_choose_k} \nonumber\\
    &=&N\sum_{k=1}^{N-1}\left(\frac{T^{k-1}}{k!}\frac{(k+3)}{2}\right),
\end{eqnarray}
where we use the fact that \(\binom{k-1}{l}\) term in Eq. (\ref{eq:P_dep_Hint_choose_k}) corresponds to the case of considering a commutator (out of the \(2^{k-1}\) possibilities) by placing \(l\) number of \(H_{int}\) out of the \(k-1\) possible places in Eq. (\ref{eq:H_F_k_comm_gen}). This commutator accordingly provides an \((l+2)\)-body factor in instantaneous power which is expressed in Eq. (\ref{eq:P_dep_Hint_choose_k}) with an extensivity contribution of factor \(N\).

{\it (B)} \(\mathbf{k\geq N}\). A similar argument is presented here as well. Here however, one caveat exists, i.e., a \(k\)-body interaction term is obtained till \(k=N-1\), after which only \(N\)-body interaction terms appear
\begin{align}
    |P(nT)|=\sum_{k=N}^{\infty}|P(nT)^{(k)}|&\leq\nonumber \sum_{k=1}^{N-1}k||H_{ch}^k||\\&\leq\nonumber\sum_{k=N}^{\infty}\frac{T^{k-1}}{2^{k-1}k!}N\left(2\binom{k-1}{0}+3\binom{k-1}{1}+4\binom{k-1}{2}+.....+N\binom{k-1}{N-2}\right)\\&\quad\quad\quad\quad+\nonumber N\left(N\binom{k-1}{N-1}+N\binom{k-1}{N-2}+N\binom{k-1}{N-3}+.....+N\binom{k-1}{k-1}\right)\\
    &=\nonumber\sum_{k=N}^{\infty} \frac{T^{k-1}}{2^{k-1}k!} \left[N\left(\sum_{i=0}^{N-2}(i+2)\binom{k-1}{i}\right) + N^2\left(\sum_{i=N-1}^{k-1} \binom{k-1}{i}\right) \right]\\
    &=\nonumber N\left[\sum_{k=N}^\infty \frac{T^{k-1}}{k!} \frac{\sum_{i=0}^{N-2}(i+2)\binom{k-1}{i}}{2^{k-1}}\right] + N^2\left[\sum_{k=N}^\infty \frac{T^{k-1}}{k!}\frac{\sum_{i=N-1}^{k-1} \binom{k-1}{i}}{2^{k-1}} \right]\\&= N\left[\sum_{k=0}^\infty \frac{T^{k+N-1}}{(k+N)!} \frac{\sum_{i=0}^{N-2}(i+2)\binom{k+N-1}{i}}{2^{k+N-1}}\right] + N^2\left[\sum_{k=0}^\infty \frac{T^{k+N-1}}{(k+N)!}\frac{\sum_{i=N-1}^{k+N-1} \binom{k+N-1}{i}}{2^{k+N-1}} \right].
\end{align}
In the last line, we change the variable \(k\) with \(k+N\) which simplifies the equation. Now, one can follow that the second term in the above equation is responsible for quadratic scaling and it provides the evidence for the non-linear scaling advantage in the long-range system. Hence the bound on instantaneous power reads
\begin{eqnarray*}
    |P(nT)|&\leq& N\left[\sum_{k=1}^{N-1}\frac{T^{k-1}}{k!}\frac{(k+3)}{2} + \sum_{k=0}^\infty \frac{T^{k+N-1}}{(k+N)!} \frac{\sum_{i=0}^{N-2}(i+2)\binom{k+N-1}{i}}{2^{k+N-1}}\right] + N^2\left[\sum_{k=0}^\infty \frac{T^{k+N-1}}{(k+N)!}\frac{\sum_{i=N-1}^{k+N-1} \binom{k+N-1}{i}}{2^{k+N-1}} \right].
\end{eqnarray*} 
\end{widetext}

{\it Note}. One of the key concerns regarding this advantage is whether it remains achievable in the thermodynamic limit. Here, we demonstrate that the  quantum advantage diminishes as the system size increases. In the second term, as \( k\) increases, the sum in the numerator also grows, contributing a nonzero effect that helps to sustain the super-extensive scaling. However, with increasing system size, the factorial term \((k+N)!\) in the denominator becomes dominant, causing the quantum advantages to wash away. This result is intuitive, as the terms corresponding to \( k \geq N \) decrease with increasing \( N \), making the linear scaling term dominant. Our analysis indicates that quantum advantages emerge only at finite system sizes which is moderately high and importantly can be experimentally achieved.\\

\begin{figure}
    \centering
    \includegraphics[width=\linewidth]{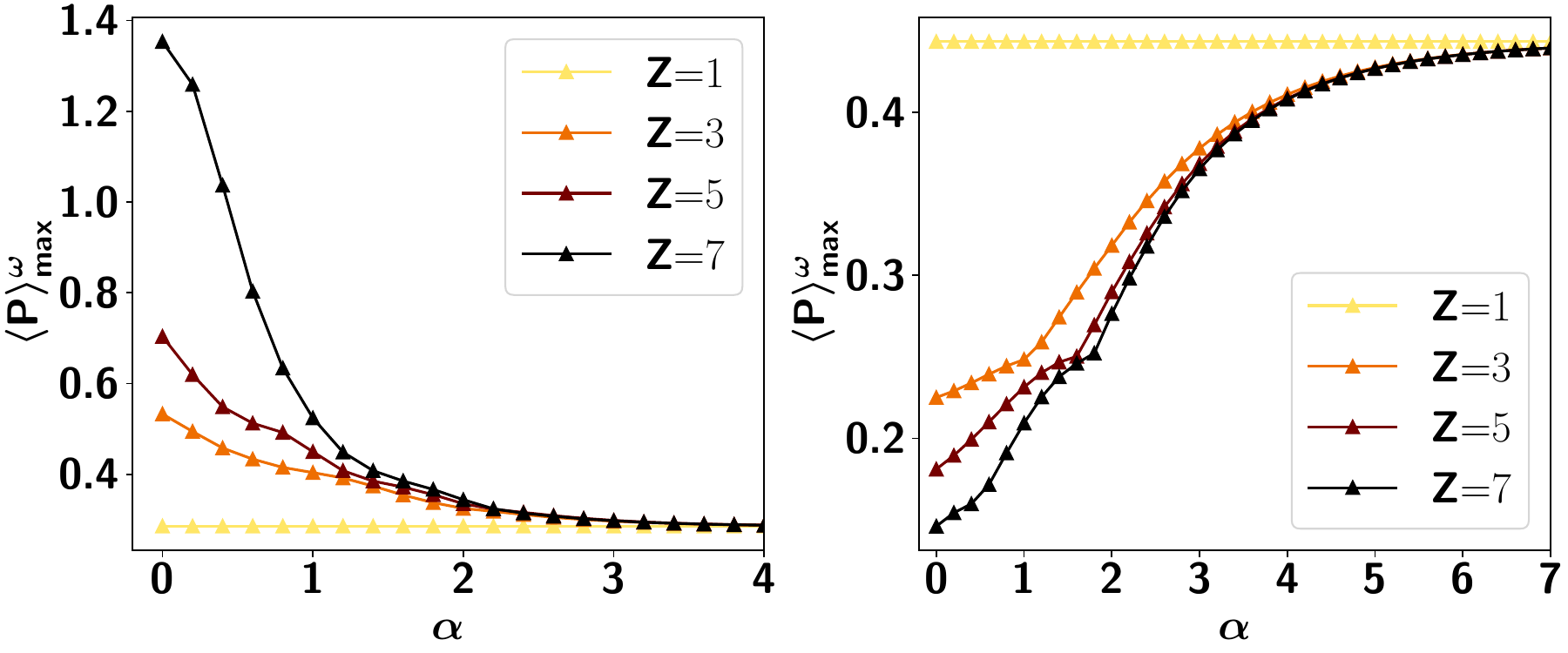}
    \caption{ $\langle P \rangle_{\max}^{\omega}$\textbf{(ordinate) vs fall-off rate \(\alpha\) (abscissa)} with \(J=20\), \(h_z=0.5\) and \(N=8\) \textbf{(left)} and \(J=0.2\), \(h_z=0.5\) and \(N=8\) \textbf{(right)}. Here, the energy spectrum of the battery Hamiltonian is normalized to be bounded by \([-1,1]\). Clearly, low \(\alpha\) with high \(Z\) provides high \(\langle P \rangle _{\max}^{\omega}\) for \(\frac{h_z}{J}\ll 1\) and high \(\alpha\) with low \(Z\) provides high \(\langle P \rangle _{\max}^{\omega}\) for \(\frac{h_z}{J}> 1\). All the axes are dimensionless.}
    \label{fig:lr_MaxPvAlpha}
\end{figure}

\begin{figure}
    \centering
    \includegraphics[width=\linewidth]{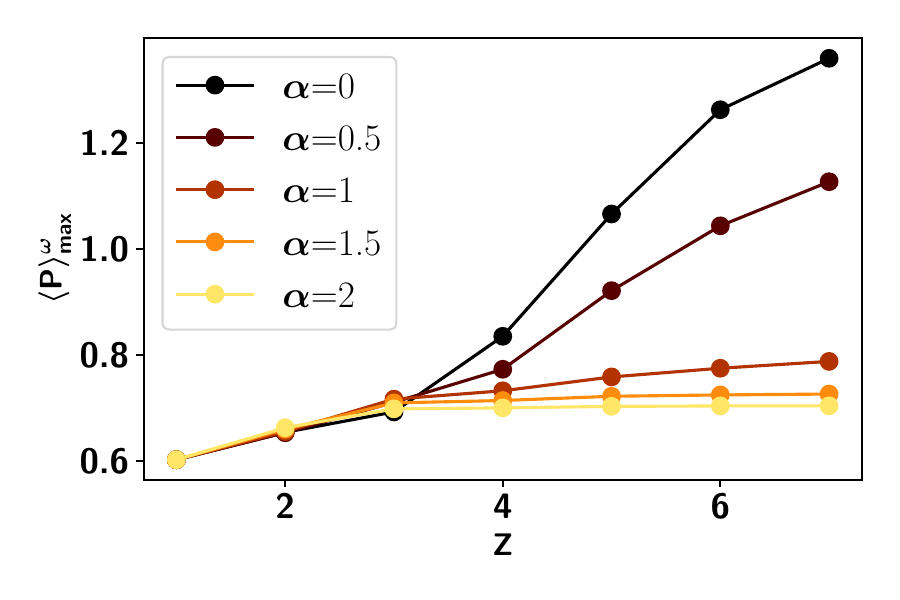}    
    \caption{ \textbf{$\langle P \rangle_{\max}^{\omega}$ (ordinate) vs coordination number \(Z\) (abscissa) with \(\gamma = -1\).} Dependency upon \(Z\) is displayed for different long-range interaction strength, \(\alpha\),  in the long-range \(XY\) model in Eq. (\ref{eq:H_int}). Other system parameters are \(J = 20\), \(h_z = 0.5\) and \( N=8\). Clearly, low \(\alpha\) and high \(Z\) provide high \(\langle P \rangle _{\max}^{\omega}\), thereby highlighting the beneficial role of LR interactions. All the axes are dimensionless.}
    \label{fig:lr_PvZ_J5}
\end{figure}

\subsection{Response of power with variable range interactions}

For physical systems like ion traps \cite{Eckardt2017RMP,iontrapp_floquet1, iontrapp_floquet2, iontrapp_floquet3}, the variable-range interactions appear naturally in the system. We use two parameters to vary the range of interactions in order to analyze the characteristic response of power calculated for a normalized work.\\
\textbf{(A) Power-law fall-off rate coefficient,} \(\alpha\).  The contrasting behavior of \(\langle P \rangle_{\max}^{\omega}\) emerges in the \(\frac{h_z}{J} \ll 1\) (paramagnetic) and \(\frac{h_z}{J} > 1\) (ferromagnetic) regimes, as depicted in Fig. \ref{fig:lr_MaxPvAlpha}. In the paramagnetic regime, where the inter-site interactions dominate, the increment in power for the truly long-range regime (\(\alpha\leq 1\)), greatly outperforms that of the short-range (\(\alpha>2\)) and quasi long-range regimes (\(1<\alpha\leq2\)), especially for highly interconnected systems \((Z \sim N-1)\). On the other hand, short-range interactions yield higher power in the ferromagnetic regime with \(Z \sim N-1\) as compared to the paramagnetic domain. By observing the higher energy scale obtained for the paramagnetic regime, one can conclude that a genuine quantum advantage is present for truly long-range interacting charging Hamiltonian in the paramagnetic phase of the model.

\textbf{(B) Coordination number,} \(Z\).
The magnitude of \( \langle P \rangle _{\max}^{\omega}\) increases with \(Z\) for different \(\alpha\leq 2\) values belonging to LR and quasi long-range regimes and in the domain where \(J\gg h_z\).  In particular, as \(\alpha\) increases, the maximum power does not change significantly with \(Z\) while for small \(\alpha\), especially when \(\alpha <1\), \(\langle P \rangle_{\max}^{\omega}\) increases monotonically with \(Z\) (see Fig. \ref{fig:lr_PvZ_J5}). This is possibly due to the fact that  strong long-range interaction is capable of creating high entanglement among all the sites, leading to a greater quantum advantage in QB in comparison with the nearest-neighbor or few-neighbor interacting charger.

\section{Advantage of having NN and NNN interactions though no super-linear scaling}
\label{sec:advNNandNNN}

Now, let us analyze whether along with nearest-neighbor interactions, denoted by \(J_1\), if one incorporates  next-nearest neighbor interaction with \(J_2\)   in the charging Hamiltonian (i.e., \(\alpha =0\) and \(Z=2\) in Eq. (\ref{eq:H_int})), any enhancement of power can be achieved or not. The Hamiltonian can be written explicitly as \begin{eqnarray}
     H_{int}^{NNN} &=& \sum_{\substack{i = 1}}^{N-1}{\frac{J_1(t)}{2}(\sigma_{i}^{x}\sigma_{i+1}^{x}} + \gamma\sigma_{i}^{y}\sigma_{i+1}^{y})  \nonumber\\&+&  \sum_{\substack{i = 1}}^{N-2}{\frac{J_2(t)}{2}(\sigma_{i}^{x}\sigma_{i+2}^{x}} + \gamma\sigma_{i}^{y}\sigma_{i+2}^{y}),
     \label{eq:H_J1_J2}
\end{eqnarray}
To address this query, we choose two paths -- (1) we study \(\langle P \rangle _{\max_n} (\omega)\) by varying the strength of \(J_2/J_1\) for a fixed \(\omega\)  and the same with the increase of \(\omega\) for different values of \(J_2/J_1\); (2) secondly,  for a large values of \(\omega\), we apply Floquet-Magnus expansion \cite{MagnusPhysreport} and study the scaling behavior of this model with \(N\) for various  \(J_2/J_1\) strength and try to see whether we can beat linear scaling obtained for NN interacting charging.

\subsection{Gain in power  with NNN interacting charger}
We establish here that although  short-range interaction can be beneficial to enhance the power of the QB,  super-extensive scaling can only be attained with LR interactions. More precisely, we focus here on the trade-off relation between nearest-neighbor and next-nearest neighbor interaction strengths present in the charging Hamiltonian which also depends on the stroboscopic time.
We first note that power can be enhanced when charging Hamiltonian contains any kinds of interactions, thereby confirming the role of quantumness for storing energy in the battery. In case of Floquet charging, another crucial component is the frequency. To determine the role of \(\omega\) and the interaction strengths, we investigate the behavior of  \(\langle P \rangle _{\max_n} (\omega, J_2/J_1, h_z)\) where \(J_2/J_1\) and \(h_z\) are system parameters of the QB and charger respectively. Since we are interested to explore the role NNN interactions in charging, we fix \(h_z\) to be moderately low compared to the interaction strength. 


\begin{figure}
    \centering
    \includegraphics[width=\linewidth]{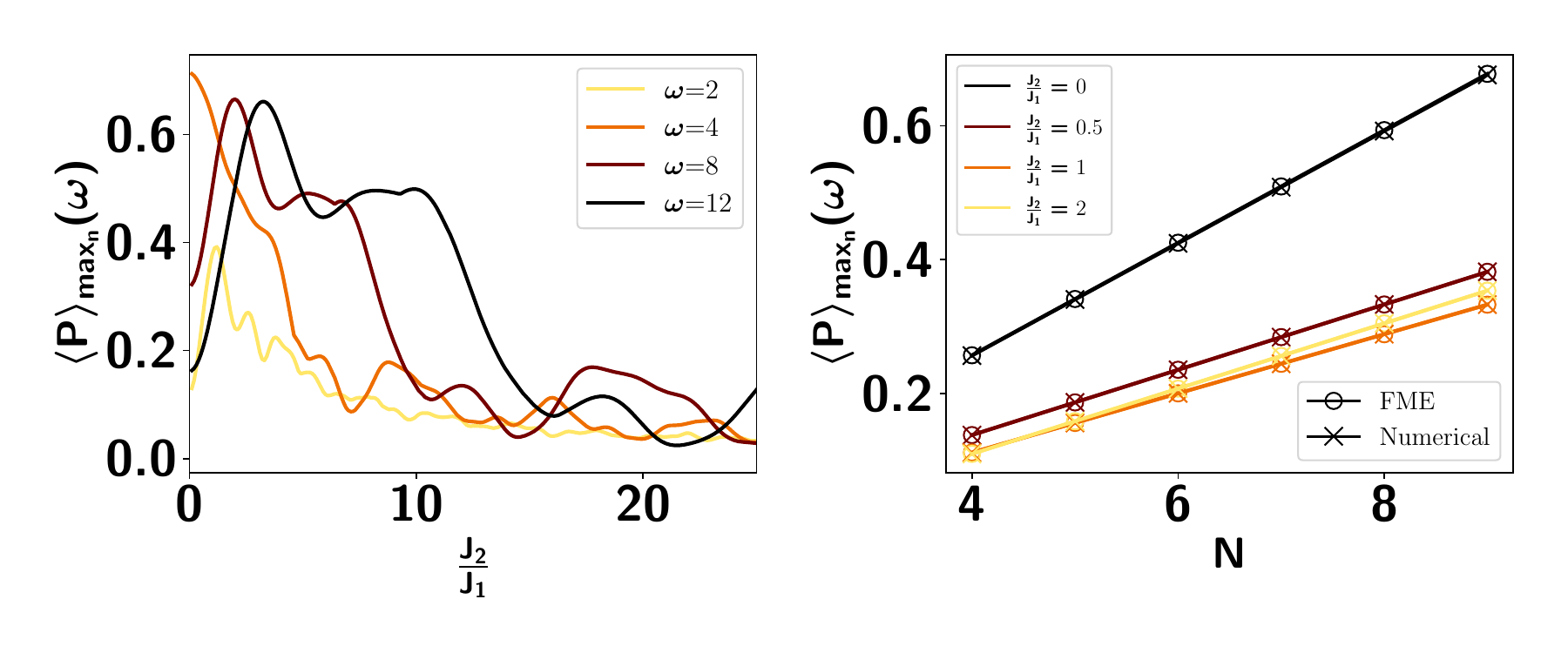}
    \caption {\textbf{(a)\(\langle P \rangle _{\max_n} (\omega)\) (ordinate)  with respect to \(\frac{J_2}{J_1}\) (abscissa) of the \(XY\) model with the increase of the frequency of the periodic driving, \(\omega\) (from light to dark lines)}.
    Clearly, there exists \(\omega_c\) above which NNN interaction combined with NN ones provides advantage over the charging Hamiltonian with only NN interactions. 
    Both the axes are dimensionless. {\bf (b) Scaling analysis for charging Hamiltonian with NN and NNN interactions.} $\langle P \rangle_{\max_{n}} (\omega)$ (vertical axis) against $N$ (horizontal axis). Here $\omega=25$ and $h_z = 0.5$. Dark to light solid lines represent  the increase of the values of \(\frac{J_2}{J_1}\). The scaling of  $\langle P \rangle _{\max_n} (\omega)$  is computed both by Floquet-Magnus expansion and by numerical methods which match exactly. It is evident from the figure that $\langle P \rangle_{\max_n} (\omega) \sim a N + b$ where \((a,b)\) depends on the \(\frac{J_2}{J_1}\) values. For example, for \(\frac{J_2}{J_1} = 0,\, (0.084,-0.079)\); \(\frac{J_2}{J_1} = 0.5,\,  (0.0487, -0.0569)\); \(\frac{J_2}{J_1}= 1,\,  (0.044,-0.065)\); and \(\frac{J_2}{J_1}=2, \,  (0.048, -0.085)\).     Both the axes are dimensionless. }
    \label{fig:NNN_MaxPvsJ2J1_merged}
\end{figure}

{\it Observation 1.} The entire profile of  \(\langle P \rangle _{\max_n} (\omega, J_2/J_1)\) depends crucially on \(\omega\) (see Fig. \ref{fig:NNN_MaxPvsJ2J1_merged}), showing the importance of frequency in the Floquet driving. It is evident from the investigation that for a fixed system parameters, the maximum of \(\langle P \rangle _{\max_n} (\omega)\) is achieved only for a single value of \(\omega\).


{\it Observation 2.}  \(\langle P \rangle _{\max_n} (\omega, J_2/J_1)\) oscillates  non-uniformly with the variation of \(J_2/J_1\) for a fixed \(\omega\) value although it saturates when \(J_2/J_1\) is moderately high. In particular,  for some values of \(\omega\), there  exist a range of \(J_2/J_1\) for which  \(\langle P \rangle _{\max_n} (\omega, J_2/J_1) \geq \langle P \rangle _{\max_n} (\omega, J_2/J_1=0) \), thereby illustrating the benefit of NNN interactions in charging. However, the increasing value of \(J_2\) over \(J_1\) is not ubiquitously beneficial as depicted in Fig. \ref{fig:NNN_MaxPvsJ2J1_merged}. It only highlights that the competition between NN and NNN interactions matters in storing energy. This observation also justifies the importance of maximizing \(\omega\) for studying the power of the QB.  

{\it Observation 3.} Let us consider the maximum stored energy with stroboscopic time, given by \(W_{\max_n} (J_2/J_1, h_z) = \underset{n}{\max}\quad W(nT)\) which reaches maximum for some \(\omega\) values when the system parameters are fixed. The  \(\omega\) for which \(W_{\max_n}\) achieves maximum changes with the ratio between the NN and NNN interactions for a weak magnetic field strength. Interestingly, we observe that with the increasing \(J_2/J_1\),  \(W_{\max_n}\) increases with \(\omega\) and as shown for \(\langle P \rangle _{\max_n}({\omega}, J_2/J_1)\),  there exists \(\omega_c\) for a fixed \(J_2/J_1\) which leads to the maximum of  \(W_{\max_n}\). 

\subsection{No benefit in scaling with NNN interactions} 

To perform the scaling analysis analytically, we will derive the time-independent Floquet Hamiltonian for the charging by using Floquet-Magnus expansion (FME) \cite{magnus1954solution, MagnusPhysreport, Bukov2015_review}. Since the charging Hamiltonian is  periodic in time, i.e., \(H_{ch}(t) = H_{ch}(t+T)\), we invoke Floquet theory to study the dynamics of the quantum battery. We calculate the time-independent Floquet Hamiltonian, \(H_{ch}^F[t_0]\), which can be used to evolve the system at stroboscopic time periods \(\tau\) (\(\tau \ \epsilon \ n(T+t_0) \ \forall \ n   \ \epsilon \, \mathbb{Z}^+\)) as
\begin{equation}
    U^F(\tau,t_0) = e^{-iH_{ch}^F[t_0]\tau}.
    \label{eq:U_Flo}
\end{equation}
\begin{widetext}

We fix $t_0 =0$, thereby neglecting  $[t_0]$ further in our calculation. To compute \(H_{ch}^F\) in the high frequency limit, i.e., when \(\omega\) is high enough, we use the Floquet-Magnus expansion upto \(\mathcal{O}(T^3)\) terms and rewrite the 
charging Hamiltonian in Eq. (\ref{eq:H_J1_J2}) as
   \begin{eqnarray}
    H^F_{ch} &= H^{F_0} + H^{F_1} + H^{F_2} + H^{F_3}, 
\end{eqnarray}
where 
\begin{eqnarray}
    H^{F_0} &=& \frac{H_1 + H_2}{2}, \,\, 
H^{F_1} = -i \frac{T}{8} [H_2, H_1], \nonumber \\
H^{F_2} &=& -\frac{T^2}{96} \big[[H_2, H_1], H_1 - H_2\big], \nonumber \\
H^{F_3} &= &i \frac{T^3}{384} \big[H_2, [[H_2, H_1], H_1]\big],
\end{eqnarray}
with

\begin{align}
H_1 &= h_z \sum_{j} \sigma_{j}^z 
+ \frac{J_1'}{2} \sum_{j} \big( \sigma_j^x \sigma_{j+1}^x + \gamma \sigma_j^y \sigma_{j+1}^y \big) + \frac{J_2'}{2} \sum_{j} \big( \sigma_j^x \sigma_{j+2}^x + \gamma \sigma_j^y \sigma_{j+2}^y \big), 
\label{eq:H1}\\
H_2 &= h_z \sum_{j} \sigma_{j}^z 
- \frac{J_1'}{2} \sum_{j} \big( \sigma_j^x \sigma_{j+1}^x + \gamma \sigma_j^y \sigma_{j+1}^y \big) - \frac{J_2'}{2} \sum_{j} \big( \sigma_j^x \sigma_{j+2}^x + \gamma \sigma_j^y \sigma_{j+2}^y \big). 
\label{eq:H2}
\end{align}
Using  Eqs. (\ref{eq:H1}) and (\ref{eq:H2}) with \(J_i = J_i'/2\) (\(i=1,2\)), we can explicitly write 
\begin{align*}
H^{F_0} &= h_z\sum_{j}\sigma_{j}^z \\
H^{F_1} &= \frac{T}{2}(1-\gamma)h_z \Bigg[J_1 \sum_{j}(\sigma_j^y \sigma_{j+1}^x +\sigma_j^x \sigma_{j+1}^y ) + J_2  \sum_j ( \sigma_j^y \sigma_{j+2}^x + \sigma_j^x \sigma_{j+2}^y )\Bigg] \\
H^{F_2} &=  - \frac{T^2}{3}(1- \gamma)h_z \Bigg[J_1 J_2 \sum_j ( \sigma_j^z \sigma_{j+1}^x \sigma_{j+2}^x -\gamma\sigma_j^z \sigma_{j+1}^y \sigma_{j+2}^y) + J_1 J_2\sum_j(\sigma_j^x \sigma_{j+1}^x \sigma_{j+2}^z - \gamma\sigma_j^y \sigma_{j+1}^y \sigma_{j+2}^z)\\ &\quad + J_1J_2\sum_j(\sigma_j^x \sigma_{j+1}^z \sigma_{j+3}^x - \gamma \sigma_j^y \sigma_{j+1}^z \sigma_{j+3}^y) + J_1 J_2\sum_j(\sigma_j^x\sigma_{j+2}^z \sigma_{j+3}^x - \gamma\sigma_j^y\sigma_{j+2}^z \sigma_{j+3}^y ) \\
&\quad + J_1^2 \sum_j( \sigma_j^x \sigma_{j+1}^z \sigma_{j+2}^x - \gamma\sigma_j^y \sigma_{j+1}^z \sigma_{j+2}^y ) + J_2^2 \sum_j( \sigma_j^x\sigma_{j+2}^z \sigma_{j+4}^x - \gamma\sigma_j^y\sigma_{j+2}^z \sigma_{j+4}^y ) +(J_1^2 + J_2^2)(1-\gamma) \sum_j \sigma_j^z \Bigg],
\end{align*}
and so on. For FME expansion, we consider \(H_{ch}^F\) instead of \(H_{ch}\) in Eq. (\ref{eq:U_num}). 
\end{widetext}


Let us now investigate how \(\langle P \rangle _{\max_n}(\omega)\) scales with \(N\) obtained analytically  using FME  and by explicit numerical simulation for a high values of \(\omega\). We find that  for a fixed \(h_z\) and a high \(\omega\), \(\langle P \rangle _{\max_n} (\omega) \sim a N + b\) where \(a\) and \(b\) are constants. We observe that  the scaling with \(N\) for different \(J_2/J_1\) values via numerical simulation exactly matches with the one obtained by using FME (see Fig. \ref{fig:NNN_MaxPvsJ2J1_merged} with \(\omega= 25\)).  It demonstrates that despite the presence of next-nearest neighbor term in the charging Hamiltonian,  the scaling cannot be made super-extensive which indicates that  SR interactions are not enough to attain super-linear scaling in power.

\section{Scaling with extended \(XY\) model}
\label{sec:extendedXY}

In order to compute the average maximum power with the charging being the extended \(XY\) model in Eq. (\ref{eq:LRextXY}), we first rewrite the Hamiltonian  in the fermionic basis following the Jordan-Wigner transformation \cite{lieb1961,barouch_pra_1970_1,barouch_pra_1970_2,glen2020} as
\begin{align}
    \sigma^x_n &=  \left( c_n + c_n^\dagger \right)
 \prod_{m<n}(1-2 c^\dagger_m c_m), \nonumber \\
    \sigma^y_n &=i\left( c_n - c_n^\dagger \right)
 \prod_{m<n}(1-2 c^\dagger_m c_m), \nonumber \\\text{and}\quad
\sigma^z_n&=1-2 c^\dagger_n  c_n,
 \label{eq:Jordan_wigner}
\end{align}
where \(c_m^\dag\)(\(c_m\)) is the creation (annihilation) operator of spinless fermions and they follow fermionic commutator algebra. The Hamiltonian in such basis reads as
\begin{eqnarray*}
   \nonumber H_{int}^{extXY}&=&\sum_{n}\sum_{r}\frac{J_r(t)}{2}[(1+\gamma)c_n^\dagger c_{n+r}\\&+&(1-\gamma)c_n^\dagger c_{n+r}^\dagger+ \text{h.c}] 
   \label{eq:JW_hamil}
\end{eqnarray*}

We prepare the initial state as the ground state of \(H_B=\frac{h_z}{2}\sum_{j}\sigma_j^z\) and it  is evolved through \(H_{ch}=H_B+H_{int}^{extXY} \) which can be mapped to a quadratic free fermionic model given in Eq. (\ref{eq:JW_hamil}). To compute the power of the battery, we rewrite the charging Hamiltonian in the momentum space by performing a Fourier transform of \(c_n=\frac{1}{\sqrt{N}}\sum_k e^{-i\phi_kn}c_k\) where \(\phi_k=\frac{(2k-1)\pi }{N}\) \(\forall k\in [1,N/2]\), \(c_k\) being the operator in the momentum basis and  in this space, the model is described as \(H_{int}^{extXY}=\underset{k\ge 0}{\bigoplus} \Psi_k H_k \Psi_k^\dagger\) where 
\begin{equation}
    H_k=\begin{bmatrix}
        \text{Re}(J_k^\alpha(t))(1+\gamma)-h_z & i\text{Im}(J_k^\alpha(t))(1-\gamma)\\
        -i\text{Im}(J_k^\alpha(t))(1-\gamma) & -\text{Re}(J_k^\alpha(t))(1+\gamma)+h_z
    \end{bmatrix},
\end{equation}
with $J_k^\alpha = \sum_{r = 1}^\frac{N}{2}J_r e^{i\phi_kr}$ and $\Psi_k^\dagger = (c_k \quad c_{-k}^\dagger )$ being the Bogoliubov basis. Hence, the effective unitary that can describe the stroboscopic evolution of QB for each mode \(k\)  is given as
\begin{eqnarray}
    U_k&=&\nonumber\exp(-iH_k(-J)\frac{T}{2})\exp(-iH_k(+J)\frac{T}{2})\\&=&\exp(-H_k^F T),
\end{eqnarray}
where \(H_k^F\) is the Floquet Hamiltonian. This Hamiltonian can be evaluated  since \(H_k^F=\sum_{k\ge 0}\beta_T(k)\vec{n}_k.\vec{\sigma}_k\) with  $\vec{n}_k=\frac{1}{\sqrt{1-u_k^0}} \begin{bmatrix}
        u^x_k, & u^y_k, & u^z_k
    \end{bmatrix}$ and $\beta_T(k)=\frac{1}{T}\cos^{-1}u_k^0$. Here $\vec{\sigma_k}$ are the Pauli matrices in the Bogoliubov basis. 
 The elements of \(\vec{n}_k\) read as
\begin{eqnarray}
    u_k^0&=&\nonumber\cos(\frac{E_k^1T}{2})\cos(\frac{E_k^2T}{2})\\&-&\nonumber\cos(\Delta_k)\sin(\frac{E_k^2T}{2})\sin(\frac{E_k^1T}{2}),\\
    u^x_k &=&\nonumber\sin(\frac{E_k^1T}{2})\sin(\frac{E_k^2T}{2})\sin(\Delta_k),\\u^y_k &=&\nonumber\cos(\frac{E_k^1T}{2})\sin(\frac{E_k^2T}{2})\sin(\theta_k^2)\nonumber\\&+&\nonumber\cos(\frac{E_k^2T}{2})\sin(\frac{E_k^1T}{2})\sin(\theta_k^1),\\u^z_k &=&\nonumber-\cos(\frac{E_k^1T}{2})\sin(\frac{E_k^2T}{2})\cos(\theta_k^2)\\&-&\cos(\frac{E_k^2T}{2})\sin(\frac{E_k^1T}{2})\cos(\theta_k^1),
\end{eqnarray}
where \(E_k^{i}=\sqrt{(\text{Re}(J_k^\alpha)_i-h_z)^2+\text{Im}(J_k^\alpha)_i^2}\), \(\cos(\theta_k^i)=(\text{Re}(J_k^\alpha)_i-h_z)/E_k^i\) 
 with \(i\in\{1,2\}\) and \(\Delta_k=\theta_k^1-\theta_k^2\). Therefore,  the energy stored in the QB at the \(n^{th}\) stroboscopic time can be analytically obtained  as
\begin{equation*}
W(nT)=\sum_{k\ge 0}\bra{0_k}(U_k^\dagger)^n H_{B,k}(U_k)^n\ket{0_k}-\bra{0_k}H_{B,k}\ket{0_k}
\end{equation*}
with \(\ket{0_k}\) is the initial state and \(H_{B,k}=-h_z\sigma^z_k\) is the battery Hamiltonian in the momentum space. The work output takes the form as
\begin{equation}
    W(nT)=\sum_{k\ge 0}2h_z(1-(n_k^z)^2)\sin^2(n\cos^{-1}u_k^0).
\end{equation}
It is evident from the above expression that the work in the stroboscopic time  depends nonlinearly on the frequency of the Floquet driving, \(T=2\pi/\omega\). Clearly, it is a function of \(\omega\), \(\alpha\), \(h_z\) and \(J\), i.e.,  \(f(\alpha, \omega, h_z, J)\) which indicates that by calibrating \(\omega\),  \(\alpha\) and  \(h_z/J\),  we can obtain nonlinear scaling in the maximum average power, \(\langle P \rangle _{\max}^\omega\) with system-size.
\bibliography{ref}

\begin{thebibliography}{96}
\bibitem [{Ome()}]{Omegamaxreason}%
  \BibitemOpen
  \href@noop {} {}\bibinfo {note} {It is clear from Eq. (\ref{eq:lmg_bound})
  that power is non-linearly dependent on the frequency range \(\omega\) and a
  maximization should be performed to saturate the bound.}\BibitemShut {Stop}%
\end{thebibliography}

\end{document}